\documentclass{svmult}
\usepackage{amsmath,amsfonts}
\usepackage{hyperref}
\usepackage{graphicx}
\usepackage[utf8]{inputenc}
\usepackage{slashed}
\usepackage{dsfont}
\usepackage{bm}
\usepackage{mathrsfs}
\usepackage{tensor}
\usepackage{geometry}
\usepackage{xcolor}
\usepackage{mathptmx}
\usepackage{helvet}
\usepackage{courier}
\usepackage{makeidx}
\usepackage{multicol}
\usepackage{footmisc}
\usepackage{appendix}
\usepackage{amssymb}
\usepackage{epic}
\usepackage{eepic}
\usepackage[matrix,arrow]{xy}
\usepackage{epsfig}
\usepackage{url}
\usepackage{hyperref}
\usepackage{multirow}
\usepackage{graphicx}				
\usepackage{amssymb}
\usepackage{amssymb}
\usepackage{color}
\usepackage{amsmath}
\usepackage{float}
\usepackage{tgtermes}
\usepackage{mathrsfs}
\usepackage{mathrsfs}
\usepackage{mathtools}
\usepackage{fullpage}
\usepackage{bm}
\newtheorem{defn}{Definition}
\newtheorem{naidefn}{Na\"{i}ve Definition}
\newtheorem{proppy}{Proposition}
\newtheorem{lemmy}{Lemma}
\newtheorem{cory}{Corollary}
\newtheorem{rmkk}{Remark}
\newtheorem{thm}{Theorem}

\usepackage{amscd}
\makeatletter
\newcommand{\dotr}[1]{%
  \mathpalette\@dotr{#1}%
}
\newcommand*{\@dotr}[2]{%
  \sbox0{$\m@th#1#2$}%
  \usebox{0}%
  \raisebox{\dimexpr\ht0-\height}{$\m@th#1\@smallbullet#1\bullet$}%
  \kern\scriptspace
}
\newcommand*{\@smallbullet}[2]{%
  \scalebox{.5}{$\m@th#1#2$}%
}
\makeatother

\newcommand{\beq}{\begin{equation}}
\newcommand{\eeq}{\end{equation}}
\newcommand{\bea}{\begin{eqnarray}}
\newcommand{\eea}{\end{eqnarray}}

\begin{document}

\title*{The Derived Category of Coherent Sheaves and B-model Topological String Theory}
\author{Stephen Pietromonaco}
\maketitle
\abstract{This elementary survey article was prepared for a talk at the 2016 Superschool on Derived Categories and D-branes.  The goal is to outline an identification of the bounded derived category of coherent sheaves on a Calabi-Yau threefold with the D-brane category in B-model topological string theory.  This was originally conjectured by Kontsevich \cite{kontsevich_homological_1994}.  We begin by briefly introducing topological closed string theory to acquaint the reader with the basics of the non-linear sigma model.  With the inclusion of open strings, we must specify boundary conditions for the endpoints; these are what we call D-branes.  After briefly summarizing the necessary homological algebra and sheaf cohomology, we argue that one should think of a D-brane as a complex of coherent sheaves, and provide a physical motivation to identify complexes up to homotopy.  Finally, we argue that renormalization group (RG) flow on the worldsheet provides a physical realization of quasi-isomorphism.  This identifies a stable object in the derived category with a universality class of D-branes in physics.  I aim for this article to be an approachable introduction to the subject for both mathematicians and physicists.  As such, it is far from a complete account.  The material is based largely on lecture notes of E. Sharpe \cite{sharpe_lectures_2003} as well as the paper \cite{aspinwall_d-branes_2004} of P. Aspinwall.}

\bibliographystyle{unsrt}

\section{Topological Closed String Theories}

The starting point for closed string topological string theories is the non-linear sigma model which studies maps $\phi: \Sigma \to X$, where $\Sigma$ is a compact, oriented Riemann surface called the `worldsheet' and we take $X$ to be a Calabi-Yau threefold, called the `target space.'  If only closed strings are present, $\Sigma$ is taken to be without boundary.  We can take local complex coordinates $(z, \bar{z})$ on $\Sigma$, and $w^{i} = \phi^{i}(z, \bar{z})$ on $X$.  We have a K{\"a}hler metric $g_{i \bar{j}}$, as well as an anti-symmetric B-field $B_{i \bar{j}}$ on $X$.  Of course, the indices here correspond to tensor components in the complex coordinates $w^{i}$.

The theory becomes \emph{topological} after performing one of two possible twists.  In what sense is the theory topological?  Such a non-linear sigma model is a two-dimensional quantum field theory defined on the fixed Riemann surface $\Sigma$.  Therefore, to say the twisted theory is topological is to say there exists a subsector of operators such that the correlation functions are independent of the metric \emph{on the worldsheet}.  It is crucial to not confuse the metric on the worldsheet with the metric on the target Calabi-Yau.  I will review the two topologically twisted models which Witten \cite{witten_mirror_1991} called the A and B models.  The A-model will depend only on the K\"{a}hler structure on $X$ while the B-model will depend only on the complex structure.  So there will indeed be partial dependence on the target space metric, the exact form of which will depend on the model under consideration.  In addition, I will define a BRST operator $Q$ (this operator will be different in the A and B models).  The physical observables of the topological subsector will consist of products of local operators, each of which is invariant under the BRST operator $Q$.  By convention, we denote the target space by $Y$ in the A-model and as $X$ in the B-model.

Let $T_{X}$ be the complexified tangent bundle of $X$, which can be decomposed as $T_{X} = T^{(1,0)}_{X} \oplus T^{(0,1)}_{X}$.  The fermions in the theory require a choice of square-root bundles $K^{1/2}$ and $\overline{K}^{1/2}$, where $K$ and $\overline{K}$ are the canonical and anti-canonical bundles on $\Sigma$, respectively.  The non-linear sigma model action is given by: (equation (2.4) in \cite{witten_mirror_1991})

\begin{equation}
S= \int_{\Sigma} d^{2}z  \bigg( \frac{1}{2} g_{i j} \partial_{z}\phi^{i} \partial_{\bar{z}} \phi^{j}+\frac{i}{2}B_{i j} \partial_{z}\phi^{i} \partial_{\bar{z}}\phi^{j} + i \psi_{-}^{\bar{i}}D_{z}\psi_{-}^{i} g_{i \bar{i}} + i \psi_{+}^{\bar{i}}D_{\bar{z}}\psi_{+}^{i} g_{\bar{i} i} + R_{i \bar{i} j \bar{j}}\psi_{+}^{i} \psi_{+}^{\bar{i}} \psi_{-}^{j} \psi_{-}^{\bar{j}}\bigg),
\end{equation}
where $R_{i \bar{i} j \bar{j}}$ is the Riemann tensor on $X$, $D_{z}$ is the $\partial$ operator on $\overline{K}^{1/2} \otimes \phi^{*}T^{(1,0)}_{X}$, arising by pulling back the holomorphic part of the Levi-Civita connection on $T_{X}$.  Likewise, $D_{\bar{z}}$ is the $\overline{\partial}$ operator on $K^{1/2} \otimes \phi^{*}T^{(1,0)}_{X}$.  The fermion fields are sections of the following bundles,

\begin{equation}
\begin{split}
& \psi_{+}^{i} \in \Gamma\big( K^{1/2} \otimes \phi^{*}T_{X}^{(1,0)}\big), \,\,\,\,\,\,\,\,\,\,  \psi_{+}^{\bar{i}} \in \Gamma\big( K^{1/2} \otimes \phi^{*}T_{X}^{(0,1)}\big), \\
&  \psi_{-}^{i} \in \Gamma\big( \overline{K}^{1/2} \otimes \phi^{*}T_{X}^{(1,0)}\big), \,\,\,\,\,\,\,\,\,\,  \psi_{-}^{\bar{i}} \in \Gamma\big( \overline{K}^{1/2} \otimes \phi^{*}T_{X}^{(0,1)}\big).
\end{split}
\end{equation}

The sigma model action above is really a worldsheet action; the integral is over two-forms on $\Sigma$.  Therefore, all of the structures described above need to be pulled back to $\Sigma$ via $\phi$, which implies that the pullback of the metric, the B-field, and the connection will all inherit $\phi$ dependence.  As mentioned, $\psi_{\pm}^{i}$, $\psi_{\pm}^{\bar{j}}$ are the fermionic fields and the bosonic fields are the local coordinates $\phi^{i}$ and $\phi^{\bar{j}}$.\footnote{Having the bosonic fields correspond to the local coordinates on a Riemannian manifold is an idea originating in `supersymmetric quantum mechanics.'}  

The supersymmetry (SUSY) transformations are generated by the four infinitesimal fermionic parameters $\alpha_{+}, \tilde{\alpha}_{+}, \alpha_{-}, \tilde{\alpha}_{-}$.  The first two are anti-holomorphic sections of $\overline{K}^{-1/2}$ and the latter two are holomorphic sections of $K^{-1/2}$.  We refer the reader to equation (2.5) in \cite{witten_mirror_1991} for the full form of the supersymmetry transformations.  Since we have four SUSY parameters, two of each chirality, we say the resulting theory has ``worldsheet $\mathcal{N}=(2,2)$ supersymmetry."

\subsection{Closed String A-Model}

Let $Y$ be the Calabi-Yau target space in the A-model.  We consider here a restricted symmetry such that $\tilde{\alpha}_{-} = \alpha_{+} =0$ and $\alpha = \alpha_{-} = \tilde{\alpha}_{+}$.  In other words, we have only one SUSY parameter which we call $\alpha$.  We now perform the first of two possible topological twists to construct the A-model topological string theory.  Consider the field $\chi \in \Gamma\big( \phi^{*} T_{X}\big)$ which projects into $\phi^{*}T_{X}^{(1,0)}$ as $\chi^{i} = \psi_{+}^{i}$ and into $\phi^{*}T_{X}^{(0,1)}$ as $\chi^{\bar{i}} = \psi_{-}^{\bar{i}}$.  We regard $\psi_{+}^{\bar{i}}$ as a $(1,0)$ form on $\Sigma$ valued in $\phi^{*}T_{X}^{(0,1)}$ and following \cite{witten_mirror_1991}, denote it as $\psi_{z}^{\bar{i}}$.  Likewise, $\psi_{-}^{i}$ is a $(0,1)$ form valued in $\phi^{*}T_{X}^{(1,0)}$, denoted $\psi_{\bar{z}}^{i}$.  The A-model SUSY transformations are 

\begin{equation} \label{eqn:SUSYAMod}
\begin{split}
& \delta \phi^{i} = i \alpha \chi^{i} \\
& \delta \phi^{\bar{i}} = i \alpha \chi^{\bar{i}} \\
& \delta \chi^{i} = \delta \chi^{\bar{i}} =0 \\
& \delta \psi_{z}^{\bar{i}} = - \alpha \partial_{z} \phi^{\bar{i}}-i \alpha \chi^{\bar{j}} \Gamma_{\bar{j} \bar{m}}^{\bar{i}} \psi_{z}^{\bar{m}} \\
& \delta \psi_{\bar{z}}^{i} = - \alpha \partial_{\bar{z}} \phi^{i}-i \alpha \chi^{j} \Gamma_{j m}^{i} \psi_{\bar{z}}^{m}
\end{split}
\end{equation}
where $\Gamma_{j m}^{i}$ is the holomorphic part of the Levi-Civita connection on the complexified tangent bundle and $\Gamma_{\bar{j} \bar{m}}^{\bar{i}}$ is the anti-holomorphic part.  Corresponding to the single SUSY parameter $\alpha$, we define the operator $Q$ to be its generator.  As such, the variation of any local operator $W$ under a SUSY transformation with parameter $\alpha$, is given by

\begin{equation}
\delta W = -i \alpha \{Q, W\}.
\end{equation}
One can show from the action that $Q^{2}=0$, on-shell.  This means that though there may be non-zero terms equated to $Q^{2}$, they will vanish if the equations of motion are satisfied.  Thus, we have a nilpotent operator $Q$ which is commonly referred to as a BRST operator.  With this in hand, we can rewrite the sigma model action as,

\begin{equation}
S = \int_{\Sigma} i\{Q, V\} - 2 \pi i \int_{\Sigma} \phi^{*}(B + i J), 
\end{equation}
where $V= 2 \pi g_{i \bar{j}}(\psi_{z}^{\bar{j}} \bar{\partial}\phi^{i} + \partial \phi^{\bar{j}}\psi_{\bar{z}}^{i})$ and $B + iJ \in H^{2}(Y, \mathbb{C})$ is the complexified K{\"a}hler form.  Given an operator $W$, we say $W$ is $Q$-closed if $\{Q,W\}=0$ and we say it is $Q$-exact if $W=\{Q,W'\}$, for some operator $W'$.  We also call a $Q$-closed operator `BRST invariant.'  We will take it as a fact that a correlation function of a $Q$-exact operator must vanish

 \begin{equation}
 \langle \{Q, W_{1}W_{2} \ldots \} \rangle =0.  
 \end{equation}
Let us assume that $W_{2}, W_{3}, \ldots$ are $Q$-closed operators, and consider the correlation function $\langle \{Q, W_{1}W_{2} \ldots\} \rangle$ for any operator $W_{1}$.  By the fact cited above, this correlation function vanishes.  Moreover, since $Q$ behaves like a differential, we can apply Leibniz' rule to get

 \begin{equation}
 0 = \langle \{Q, W_{1}W_{2} \ldots\}\rangle = \langle W_{1} \{Q, W_{2}W_{3} \ldots\} \rangle  + \langle \{Q, W_{1}\} W_{2} W_{3} \ldots \rangle.
 \end{equation}
Since $W_{2}, W_{3}, \ldots$ are $Q$-closed operators, the term $ \langle W_{1} \{Q, W_{2}W_{3} \ldots\} \rangle$ will vanish when expanded using Leibniz' rule.  All that remains is the correlation function $\langle \{Q, W_{1}\} W_{2} W_{3} \ldots \rangle$ involving one $Q$-exact operator and the rest, $Q$-closed.  Since the original correlation function vanished, clearly this one must too.  Therefore, the presence of even one $Q$-exact operator annihilates the correlation function.  In the topological subsector, the physical observables are products of local operators, all of which are $Q$-closed (i.e. BRST invariant).
 
We note that a shift in the action by a $Q$-exact operator $S \to S + \int_{\Sigma} \{Q, S'\}$ will leave all correlation functions invariant.  In the sigma model action, the only place the complex structure of $Y$ appears is in the term $V$.  If we deform the complex structure $V \to V+\delta V$, this leads to a deformation of the action $S \to S + \int_{\Sigma} \{Q, \delta V\}$, which will leave all physical observables invariant.  Thus, it appears that the A-model topological field theory is independent of the complex structure on $Y$.  Clearly, it explicitly depends on the K{\"a}hler structure on the target space, through the term $2 \pi i \int_{\Sigma}(B + iJ)$.  

By the SUSY transformations (\ref{eqn:SUSYAMod}) we have $\delta \chi^{i} = \delta \chi^{\bar{i}}=0$, where $\chi^{i}$ and $\chi^{\bar{i}}$ are the fermionic superpartners of $\phi_{i}$ and $\phi^{\bar{i}}$, respectively.  This means the operators $\chi^{i}$ and $\chi^{\bar{i}}$ are $Q$-closed.  Thus, we have a basis of local BRST invariant operators on $\Sigma$, which we can use to write a general operator as

\begin{equation}
W_{a} = a_{I_{1} \cdots I_{p}} \chi^{I_{1}} \cdots \chi^{I_{p}},
\end{equation}
where here the capital $I_{q}$ denotes unbarred indices, and

\begin{equation}
a = a_{I_{1} \cdots I_{p}} d \phi^{I_{1}} \cdots d\phi^{I_{p}},
\end{equation}
is a $p$-form on $Y$.  By computing the variation of the operator $W_{a}$, we find that $\{Q, W_{a}\} = -W_{da}$, with an important conclusion: 

\begin{center}
\textbf{A local operator $\bm{W_{a}}$ is $Q$-closed (BRST invariant) if and only if $\bm{da=0}$.  In other words, we can identify the $Q$-cohomology in the A-model with the de Rham cohomology $\bm{H^{*}(Y,\mathbb{C})}$ on the target space.  Notice this is consistent with the A-model being independent of the complex structure on $\bm{Y}$.}
\end{center}

A correlation function in the closed string A-model is given by the following path integral,

\begin{equation}
 \langle W_{a} W_{b} \cdots \rangle = \int \mathcal{D} \phi \mathcal{D} \psi \mathcal{D}\chi e^{-S} W_{a} W_{b} \cdots.
\end{equation}
Here, we will focus just on the bosonic map $\phi: \Sigma \to Y$.  It turns out that in the topological sector, we want to restrict to maps such that the term $\{Q, V\}$ in the action vanishes.  Looking at the form of $V$, we see that we must insist $\bar{\partial} \phi^{i} = \partial \phi^{\bar{i}}=0$, i.e. $\phi$ is a holomorphic map.  So instead of performing the path integral over \emph{all} maps, we localize to only the holomorphic ones.  In this context, such a holomorphic map is called a \emph{worldsheet instanton}.  We can consider the degree-$d$ worldsheet instantons and their moduli space $\mathcal{M}_{d}$.  For example, a degree-0 map simply sends all of $\Sigma$ to a point in $Y$, implying $\mathcal{M}_{0} =Y$.  We get the following reduction of the path integral

\begin{equation}
\int \mathcal{D} \phi \mathcal{D} \psi \mathcal{D}\chi \longrightarrow \sum_{d} \int_{\mathcal{M}_{d}} (\mathcal{D}\phi)_{d} \int \mathcal{D}\psi \mathcal{D}\chi.
\end{equation}

Since the relevant space of operators in the A-model is identified with the de Rham cohomology $H^{*}(Y, \mathbb{C})$, there is a natural grading by the degree of the forms.  In physics, this is called the \emph{ghost number}, meaning if $a \in H^{p}(Y, \mathbb{C})$, then the operator $W_{a}$ is said to have ghost number $p$.  One should imagine the worldsheet instantons to be ``wrapped" on the two-cycles in $Y$.  Roughly speaking, this explains the dependence of the A-model on the K\"{a}hler structure of $Y$, as the K\"{a}hler classes control relative volumes of the two-cycles.  As noted above, the A-model is independent of the complex structure.

\subsection{Closed String B-Model}

If we perform the opposite twist we get the closed string B-model where certain fields are simply sections of different bundles over $\Sigma$.  For purposes of anomaly cancellation, we will take $c_{1}(X)=0$, i.e. take the target space to be Calabi-Yau.  Define the following combinations of the fermionic fields, $\eta^{\bar{j}} = \psi_{+}^{\bar{j}} + \psi_{-}^{\bar{j}}$, and $\theta_{j} = g_{j \bar{k}}(\psi_{+}^{\bar{k}} - \psi_{-}^{\bar{k}})$ where now the fermionic fields are sections of the following bundles

\begin{equation}
\psi_{\pm}^{\bar{i}} \in \Gamma\big( \phi^{*} T_{X}^{(0,1)}\big), \,\,\,\,\,\, \psi_{+}^{i} \in \Gamma\big( K \otimes \phi^{*} T_{X}^{(1,0)}\big), \,\,\,\,\,\, \psi_{-}^{i} \in \Gamma\big( \overline{K} \otimes \phi^{*} T_{X}^{(1,0)}\big).
\end{equation}
Let $\rho^{i}$ be a one-form on $\Sigma$ valued in $\phi^{*}T_{X}^{(1,0)}$ whose $(1,0)$ part is $\psi_{+}^{i}$ and $(0,1)$ part is $\psi_{-}^{i}$.  The B-model SUSY transformations are,

\begin{equation} \label{eqn:SUSYBMod}
\begin{split}
& \delta \phi^{i} = 0 \\
& \delta \phi^{\bar{i}} = i \alpha \eta^{\bar{i}} \\
& \delta \eta^{\bar{i}} = \delta \theta_{i} =0 \\
& \delta \rho^{i} = - \alpha d \phi^{i}.
\end{split}
\end{equation}
The physical local observables are again given by products of BRST invariant fields,

\begin{equation}
W_{A} = A_{\bar{k}_{1} \ldots \bar{k}_{q}}^{j_{1} \ldots j_{p}} \eta^{\bar{k}_{1}} \ldots \eta^{\bar{k}_{q}} \theta_{j_{1}} \ldots \theta_{j_{p}}.
\end{equation}
Clearly, such an object is a $(0,q)$-form, valued in the bundle $\bigwedge^{p}T_{X}^{(1,0)}$.  Analogously to the A-model, we find that 

\begin{equation}
\{Q,W_{A}\} = - W_{\bar{\partial}A}.
\end{equation}

\begin{center}
\textbf{In other words, in the B-model the $Q$-cohomology is the Dolbeault cohomology on the target space $\bm{H^{0,q}(X, \bigwedge^{p}T_{X}^{(1,0)})}$, with forms valued in an exterior power of the holomorphic tangent bundle. } 
\end{center}

Also, like the A-model, the path integral localizes to only certain maps $\phi$, but in this case the condition is that $\bar{\partial}\phi^{\bar{k}} = \partial \phi^{\bar{k}} = 0$.  This can only be satisfied if $\phi$ is a constant map from the worldsheet into $X$.  Clearly, the moduli space of such maps is simply $\mathcal{M}_{0} = X$.  The upshot of this is that physical observables in the B-model are given simply by \emph{ordinary} integrals over the target space.  These are essentially the \emph{period integrals} over the non-vanishing holomorphic $(3,0)$-form $\Omega$.  When considering mirror symmetry, people often say something like, ``a hard computation on one side can be converted to a trivial computation on the other side."  This idea applies here: on the A-side, correlation functions require a sum of integrals over non-trivial moduli spaces, while on the B-side, the computation reduces to simply period integrals.  These period integrals are indicative of the dependence of the B-model on the complex structure of $X$ as well as the independence of the K\"{a}hler structure.

\subsection{Topological Field Theory vs. Topological String Theory?}

It is a good time to rectify a common confusion between \textit{topological field theories} and \textit{topological string theories}.  Simply put, we take a topological field theory to be a field theory such that there exists a subsector where the correlation functions are independent of the metric on the spacetime; in our case, the string worldsheet.  The correlation functions are then given by a path integral over the bosonic fields $\phi^{i}$ as well as the fermionic fields, described in the previous section.  However, we only implicitly mention a fixed metric $h_{\alpha \beta}$ on the string worksheet $\Sigma$ itself.  We certainly do not allow for dynamics of $h_{\alpha \beta}$, as it is not summed over in the path integral.  Topological string theory arises from including the worldsheet metric as a dynamical field, which we include in the path integral prescription for correlation functions.  We describe this as ``coupling a topological field theory to worldsheet gravity."  Thus, our correlation functions now involve a sum over the genus $g$ of $\Sigma$, as well as an integral over the moduli space of complex structures on $\Sigma$.  This should come as no surprise, since string theory is a theory of quantum gravity.  Indeed, quantum gravity is by definition a quantum field theory where the metric on spacetime (in this case, the worldsheet) is dynamical and included in the path integral.  The mathematically rigorous foundation of topological string theory is known as \textit{Gromov-Witten theory}.

\section{The Open String B-Model}

With the closed string theory in hand, we now endeavor to include open strings in the theory.  This simply amounts to allowing the worldsheet $\Sigma$ to have a boundary, denoted $\partial \Sigma$.  These worldsheet boundaries have the interpretation of open string endpoints.  Under the map $\phi: \Sigma \to X$, the image of $\partial \Sigma$ is required to live on certain special submanifolds of $X$ called \emph{D-branes}.  One should interpret the D-branes as providing boundary conditions on the open string endpoints: the endpoints are forced to lie \emph{on} the D-brane (Dirichlet boundary conditions), while they are allowed to move freely \emph{within} the D-brane itself (Neumann boundary conditions).  In more physical language, we say that D-branes are non-perturbative solutions of an effective field theory.  Interestingly enough, these non-perturbative solutions were actually expected for quite a long time.  However, the true magic of their discovery \cite{dai_new_1989} is that they allow for a two-dimensional analysis, via the open string worldsheet.  This was quite exciting and unexpected.  In other words, we expected some non-perturbative solutions to exist, but had no idea these objects would support open string endpoints.

We should immediately exorcise any confusions about the distinction between \textit{boundaries} of $\Sigma$ and \textit{punctures} in $\Sigma$.  With worldsheets involving only closed strings, the strings themselves are represented by ``loops" stretching out to the infinite past or future.  Using the conformal invariance of the worldsheet theory, we can map these to simply point-like punctures on the surface of $\Sigma$.  In the path integral prescription, these punctures are superficially filled in to give a compact Riemann surface, at the expense of inserting a vertex operator at that point, representing the closed string state.  Genuine \textit{boundaries} of $\Sigma$ are different, however.  A boundary component of $\Sigma$ is superficially partitioned by punctures.  These punctures represent open string states stretching out to the infinite past or future, while the remaining segments of the boundary component are precisely what we think of as the open string endpoints ``moving in time."

Let $X$ be a Calabi-Yau threefold.  To the roughest approximation, a \textit{Dp-brane} in the context of topological string theory is a real $p$-dimensional submanifold of $X$, i.e. a representative of a class in $H_{p}(X, \mathbb{Z})$.  The convention in topological string theory is that a Dp-brane has $p$ real, spatial dimensions in the Calabi-Yau and any number of dimensions in the non-compact spacetime.   

A D-brane however, is much more than just a submanifold.  As introduced above, D-branes support open string endpoints.  Hence, these open string endpoints appear as ``particle worldlines" in the $(p+1)$-dimensional worldvolume of the Dp-brane.  Indeed there are good physical reasons to interpret this as the D-brane giving rise to a quantum field theory or gauge theory on its worldvolume.  In the context of topological strings, we ignore the time direction and consider a gauge theory on simply the $p$-dimensional subspace of $X$.  In a gauge theory on a spacetime $Z$, the physical fields are connections on, or sections of a vector bundle associated to a principal bundle defined on $Z$.  Since the endpoints of open strings appear as gauge-theoretic particles in the D-brane, we are inclined to consider a D-brane as a submanifold along with a vector bundle supported on it.  In the B-model, the objects are holomorphic, so we take the bundles to be holomorphic.  Therefore as a first pass, we make the following na\"{i}ve definition of a D-brane:

\begin{naidefn}
A single Dp-brane in the B-model topological string theory, for $p=0, 2, 4, 6$ is a complex dimension $p/2$ holomorphic submanifold $Z$ of a Calabi-Yau threefold $X$ along with a holomorphic line bundle $L\to Z$.
\end{naidefn}

It will soon become apparent that a stack of multiple D-branes will correspond to certain stable higher rank bundles.  We would like to build the category of B-model D-branes such that the objects are defined on the ambient Calabi-Yau $X$.  Under the natural inclusion $Z \hookrightarrow X$ we can pushforward holomorphic vector bundles to sheaves on $X$.  Clearly such a pushforward is not a holomorphic vector bundle on $X$: vector bundles always have sections on small enough open sets, whereas this pushforward has no sections on any open set outside $Z$.  We must broaden our consideration from merely the geometrical category of holomorphic vector bundles to the algebraic or sheaf-theoretic category of coherent sheaves.  As we will see later, we actually must further enlarge our category.  We will be compelled to understand B-model D-branes as \emph{complexes} of coherent sheaves, modulo various equivalences.  To explain these ideas we introduce now some of the required algebraic geometry.

\subsection{Coherent Sheaves and D-branes}

For some of the foundational algebraic geometry to follow, I refer the reader to \cite{hartshorne_algebraic_1997, griffiths_principles_2014}.  Let $X$ be a compact, smooth complex manifold, or more generally a scheme, with $\mathcal{O}_{X}$ its structure sheaf of regular functions.  We begin by defining a sheaf-theoretic generalization of the notion of a module over a ring.  This is known as an $\mathcal{O}_{X}$-module, and is the largest category of sheaves we will need to consider.  It contains as subcategories the coherent sheaves and locally-free sheaves, which we will introduce shortly.

\begin{defn}
For $\mathscr{E}$ a sheaf on $X$, we say $\mathscr{E}$ is an $\mathcal{O}_{X}$-module, if for all open sets $U \subseteq X$, the sections $\mathscr{E}(U)$ constitute an $\mathcal{O}_{X}(U)$-module.  In addition, the restriction morphisms must be compatible with the module structure, in the following sense: consider nested open sets $V \subseteq U$ and define sections $f \in \mathcal{O}_{X}(U)$, $s \in \mathscr{E}(U)$.  We require that $(f \cdot s)|_{V} = f|_{V} \cdot s |_{V}$, where we denote the restriction morphism as the familiar function restriction.   
\end{defn}

Notice that $\mathcal{O}_{X}$-modules are a generalization of modules over a ring.  The intrinsic geometry of $X$ gives rise to the structure sheaf $\mathcal{O}_{X}$ which naturally assigns a ring $\mathcal{O}_{X}(U)$ to each open set.  It is precisely this ring of local functions which provides the multiplication, turning $\mathscr{E}(U)$ into an $\mathcal{O}_{X}(U)$-module.  Hence, an $\mathcal{O}_{X}$-module is really a sheaf of modules.  The $\mathcal{O}_{X}$-modules constitute an abelian category.  This should come as no surprise given that abelian categories are in some sense modeled on the category of modules over a ring.

Trivially, $\mathcal{O}_{X}$ itself is an $\mathcal{O}_{X}$-module.  More generally, $\mathcal{O}_{X}^{\oplus N}$ is an $\mathcal{O}_{X}$-module called `the free $\mathcal{O}_{X}$-module of rank $N$.'  A particularly refined subcategory of $\mathcal{O}_{X}$-modules is those which look locally like $\mathcal{O}_{X}^{\oplus N}$ for some $N$.  This leads to the following definition, which will allow us to identify certain special $\mathcal{O}_{X}$-modules with holomorphic vector bundles.

\begin{defn}
A sheaf $\mathscr{E}$ on $X$ is called locally-free of rank $N$ if there exists an open cover $\{U_{\alpha} \}$ of $X$ such that $\mathscr{E}(U_{\alpha}) \cong \mathcal{O}_{X}(U_{\alpha})^{\oplus N}$.   
\end{defn}

One can show that locally-free sheaves of rank $N$ form a category.  Given that vector bundles trivialize over special open sets, locally-free sheaves seem to correspond exactly to holomorphic vector bundles.  The correspondence is made precise by the following Proposition.

\begin{proppy}
There exists a one-to-one correspondence between holomorphic vector bundles of rank $N$ on $X$ and locally-free sheaves of rank $N$ on $X$.   
\end{proppy}

\begin{proof}
The proof here is very elementary, and we only sketch it.  Given a holomorphic vector bundle $E$ on $X$, for all open sets $U$, define $\mathscr{E}(U)$ to be the sections of the vector bundle over $U$.  Since the vector bundle must trivialize, this resulting sheaf will of course be locally-free.  Conversely, given a locally-free sheaf $\mathscr{E}$, using the given isomorphism $\mathscr{E}(U_{\alpha}) \simeq \mathcal{O}_{X}(U_{\alpha})^{\oplus N}$, we can define holomorphic transition functions, which will produce a holomorphic vector bundle $E$.  
\end{proof}
\noindent Given holomorphic vector bundles $E$ and $F$, we will usually denote their corresponding locally-free sheaves by $\mathscr{E}$ and $\mathscr{F}$, respectively.

\subsection*{D6-branes and Locally-free Sheaves}

In topological string theory on a Calabi-Yau threefold $X$, when we talk about ``space-filling branes" we mean a D6-brane whose underlying homology class is a multiple of the fundamental class of $X$.  Quite simply, D6-branes are in one-to-one correspondence with locally-free sheaves on $X$.  This provides a translation between a precise mathematical notion and a phrase appearing frequently in the physics literature:

\begin{center}
\textbf{A stack of $\bm{N}$ D6-branes wrapping a Calabi-Yau threefold $\bm{X}$ corresponds to a rank $\bm{N}$ locally-free sheaf on $\bm{X}$.}
\end{center}

\noindent On a D6-brane, we specify purely Neumann boundary conditions, which allow the open string endpoint to move freely within $X$.  This choice corresponds to the constraint

\begin{equation}
\theta_{j} = g_{j \bar{k}}(\psi_{+}^{\bar{k}} - \psi_{-}^{\bar{k}})=0.
\end{equation}
Like we saw in the brief analysis of the closed string $B$-model, the BRST operator $Q$ is taken to be the Dolbeault operator $\bar{\partial}$, and we only take our local operators on the worldsheet to consist of $Q$-closed local operators.  Recalling the SUSY transformations (\ref{eqn:SUSYBMod}), these are precisely $\theta_{j}$ and $\eta^{\bar{j}}$.  But the space-filling condition forces the $\theta_{j}$ to vanish, so our local operators will only depend on $\eta^{\bar{j}}$, and of course $\phi$.  Thus, since $\bar{j}$ is an anti-holomorphic index, we conclude that our local operators must be $(0,q)$-forms, possibly valued in some bundle.

Let us attempt to construct a well-defined D-brane category, assuming at first that the only objects are D6-branes.  By the above correspondence, the objects are simply given by a bundle $E \to X$.  To give a pair of objects, is to give a pair of bundles on $X$, $E_{1} \to X$ and $E_{2} \to X$.  Since these are bundles over the same base manifold, we can define $\rm{Hom}(E_{1}, E_{2})$ to be the bundle morphisms between them.  It will be useful to note here that $\textnormal{Hom}(E_{1}, E_{2}) \simeq E_{1}^{*} \otimes E_{2}$ is itself a vector bundle with fiber defined as $\textnormal{Hom}(E_{1}, E_{2})(x) = \textnormal{Hom}(E_{1}(x), E_{2}(x))$, for all $x \in X$.  

We then take our local operators representing an open string state to be $W_{A}$, where $A$ is a $(0,q)$-form valued in the bundle $\textnormal{Hom}(E_{1}, E_{2})$.  Therefore, it is natural to define the morphisms from $E_{1} \to X$ to $E_{2} \to X$ (equivalently the open string states stretching from one D6-brane to the other), to be the Dolbeault cohomology group

\[H^{0,q}_{\bar{\partial}}\big(X, \textnormal{Hom}(E_{1}, E_{2})\big).\]
And by the familiar $\check{C}$ech-Dolbealt isomorphism, the Dolbeault cohomology group above is isomorphic to $\check{C}$ech cohomology

\begin{equation} \label{eqn:openstringst}
H^{0,q}_{\bar{\partial}}\big(X, \textnormal{Hom}(E_{1}, E_{2})\big) \simeq \check{H}^{q}\big(X, \mathscr{H}\textnormal{om}(\mathscr{E}_{1}, \mathscr{E}_{2})\big),
\end{equation}
where $\mathscr{E}_{1}$ and $\mathscr{E}_{2}$ are the locally-free sheaves corresponding to the vector bundles $E_{1}$ and $E_{2}$.  In the B-model, specifically in the case of space-filling branes, we can unambiguously assign a `ghost number' $q$ to an open string.  We will see that this will be less clean when considering branes of non-zero codimension.  

As a simple example, we can compute a three-point correlator of open string states \cite{aspinwall_d-branes_2004}.  Consider three D6-branes corresponding to holomorphic vector bundles $E_{1}$, $E_{2}$, and $E_{3}$.  Let us call the three local operators $W_{A}$, $W_{B}$, and $W_{C}$, where

\begin{equation}
A \in H^{0,1}_{\bar{\partial}}\big(X, \textnormal{Hom}(E_{1}, E_{2})\big), \,\,\,\,\, B \in H^{0,1}_{\bar{\partial}}\big(X, \textnormal{Hom}(E_{2}, E_{3})\big), \,\,\,\,\,C \in H^{0,1}_{\bar{\partial}}\big(X, \textnormal{Hom}(E_{3}, E_{1})\big).
\end{equation}
Recall that in the B-model, since instantons are suppressed, the correlation functions are given simply by integrals over $X$.  Indeed, the path integrals in the topological sector include only contributions from the moduli space $\mathcal{M}_{0}$ of degree zero harmonic maps into $X$.  But of course, these are simply constant maps, and $\mathcal{M}_{0} = X$.  This implies,

\begin{equation}
\langle W_{A} W_{B} W_{C} \rangle = \int_{X} \textnormal{Tr}(A \wedge B \wedge C) \wedge \Omega.
\end{equation}
The `integrand' is a $(3,3)$-form, which is natural to integrate over a threefold.  Of course, when wedging forms valued in the bundle $\textnormal{Hom}(E_{i}, E_{j})$, we implicitly compose the morphisms.

\subsection*{The Mukai Vector and D-brane Charges}

We have seen that when considering only D6-branes on a threefold, it sufficed to model them as objects in the category of locally-free sheaves.  The goal of this section is to gently acquaint the reader with some of the more general coherent sheaves needed to formalize D4, D2, and D0-branes.  For a rigorous definition of coherent sheaves, see \cite{hartshorne_algebraic_1997}.  For my purposes, it will suffice to think very roughly of coherent sheaves as the minimal, full abelian category arising as the ``completion" of the category of locally-free sheaves upon adding all kernels and cokernels.  

\begin{naidefn}
In the B-model topological string on a Calabi-Yau threefold $X$, a D-brane corresponds to a stable\footnote{One can use either slope stability or Gieseker stability, but I will omit discussions of stability here.} coherent sheaf $\mathcal{F}$ on $X$.  The support of the sheaf $\text{supp}(\mathcal{F})$ defines the underlying homology class of the D-brane.  
\end{naidefn}

\noindent Branes need not be pure dimensional.  For example, a coherent sheaf $\mathcal{F}$ can be supported on curves and points.  We interpret such an $\mathcal{F}$ as a \emph{bound state} of D0-D2 branes.  Such brane configurations occur, for example, in Donaldson-Thomas theory.  A helpful device for guiding intuition here is the \emph{Mukai vector} or equivalently, the \emph{D-brane charges} \cite{harvey_algebras_1998} associated to a coherent sheaf.

\begin{defn}
Let $X$ be a smooth $n$-dimensional variety and let $\mathscr{F}$ be a coherent sheaf on $X$.  The Mukai vector is defined to be
\begin{equation}
v(\mathscr{F}) = \text{ch}(\mathscr{F}) \sqrt{\text{td}(X)} = (v_{0}, \ldots, v_{n}) \in H^{2*}(X, \mathbb{Q}).
\end{equation}
If $X$ is also projective, then the D-brane charge is given simply by the Poincar\'{e} dual of the Mukai vector\footnote{In \cite{halverson_perturbative_2015}, the authors introduce `gamma classes' which encode corrections to the factor of $\sqrt{\text{td}X}$.}

\begin{equation}
\mathcal{Q}(\mathscr{F}) = \text{PD}\big( \text{ch}(\mathcal{F})\sqrt{\text{td}(X)} \big)  \in H_{2*}(X, \mathbb{Q}).
\end{equation}
By convention, we order the charges as $\mathcal{Q}(\mathscr{F})= (\mathcal{Q}_{n}, \ldots, \mathcal{Q}_{0})$, where $\mathcal{Q}_{i} \in H_{2i}(X, \mathbb{Q})$ is called the D$2i$-charge. 
\end{defn}

Recall that based on the na\"{i}ve definition, we concluded that the coherent sheaves which most directly correspond to physical D-branes are pushforwards of holomorphic vector bundles along inclusions.\footnote{This is not \emph{quite} true.  Due to a phenomenon related to the \emph{Freed-Witten anomaly}, one must also tensor by $K_{Z}^{-1/2}$ where $K_{Z}$ is the canonical bundle of $Z$.  There is a nice discussion of this in \cite{sharpe_lectures_2003,aspinwall_d-branes_2004}.}  Let $X$ be an $n$-dimensional smooth, projective variety and let $\iota: Z \hookrightarrow X$ be the inclusion of the $m$-dimensional subvariety $Z$ into $X$.  In addition, let $E$ be a rank $N$ holomorphic vector bundle on $Z$.

\begin{lemmy}
Given $X$, $Z$, and $E$ as described above, we have
\begin{equation}
\begin{split}
&\,\,\,\,\, \text{ch}_{k}(\iota_{*}E) =0 , \,\,\,\,\,\, \text{for all} \,\, k<n-m, \\
&\text{PD}\big(\text{ch}_{n-m}(\iota_{*}E)\big) = N[Z] \in H_{2m}(X, \mathbb{Q}).
\end{split}
\end{equation}
\end{lemmy}
\begin{proof}
This is a straightforward computation which can be found, for example, in \cite{fulton_intersection_1998}.  
\end{proof}
\noindent This simple result about the Chern character of pushforwards of vector bundles, immediately implies the following corollary about the D-brane charges.  

\begin{cory} \label{cory:DbrCh}
Again given $X$, $Z$, and $E$ as above, the D-brane charges satisfy
\begin{equation}
\begin{split}
&\,\,\,\,\, \mathcal{Q}_{k}(\iota_{*}E) =0 , \,\,\,\,\,\, \text{for all} \,\, k > n-m, \\
& \mathcal{Q}_{m}(\iota_{*}E)  = N[Z]  \in H_{2m}(X, \mathbb{Q}).
\end{split}
\end{equation}
\end{cory}
\begin{proof}
Using the Lemma, it follows that $\big(\text{ch}(\iota_{*}E)\sqrt{\text{td}(X)}\big)_{k}=0$ for all $k<n-m$.  Poincar\'{e} dualizing, this shows that all entires in the D-brane charge vanish for $k>n-m$, thus proving the first claim.  Note that $\big(\sqrt{\text{td}(X)}\big)_{0}=1$, and so

\begin{equation}
\big(\text{ch}(\iota_{*}E)\sqrt{\text{td}(X)}\big)_{n-m}=\text{ch}_{n-m}(\iota_{*}E).
\end{equation}
\noindent By Poincar\'{e} dualizing and applying the Lemma once more, the second claim follows.  
\end{proof}

\noindent This corollary provides a precise mathematical translation of a phrase, prevalent in the physics literature, generalizing one made earlier about D6-branes and locally-free sheaves:

\begin{center}
\textbf{In physics, one often hears about ``a stack of $\bm{N}$ D-branes wrapping a holomorphic cycle $\bm{Z \subseteq X}$."  Mathematically, this corresponds to a rank $\bm{N}$ holomorphic vector bundle on $\bm{Z}$.}
\end{center}

Let us introduce now a few of the familiar coherent sheaves one might encounter on a Calabi-Yau threefold $X$.  It has been previously observed that D6-branes correspond to locally-free sheaves.  In non-zero codimension, D4, D2, and D0-branes correspond to \emph{torsion sheaves}.  A torsion sheaf is a coherent sheaf $\mathscr{F}$ of \emph{rank zero}, which is encoded into the Mukai vector as $v_{0}=0$, or equivalently into the D-brane charges as $\mathcal{Q}_{3}=0$.  

Let $Z$ be a holomorphic subvariety of $X$.  This gives rise to a short exact sequence

\begin{equation} \label{eqn:exseqidsheaf}
0 \to \mathcal{I}_{Z} \to \mathcal{O}_{X} \to \mathcal{O}_{Z} \to 0,
\end{equation}
where $\mathcal{O}_{Z}$ is the structure sheaf on $Z$ and $\mathcal{I}_{Z}$ is called an \emph{ideal sheaf}.  In algebraic geometry, an ideal sheaf on $X$ is a rank one torsion-free sheaf $\mathcal{I}_{Z}$ with trivial determinant.  There is necessarily an injective sheaf morphism $\mathcal{I}_{Z} \to \mathcal{O}_{X}$ and the cokernel defines a subscheme $Z \subseteq X$ along with the short exact sequence above.  If $Z$ is a divisor, then $\mathcal{I}_{Z}$ is actually a line bundle, and $\mathcal{O}_{Z}$ is an example of a D4-brane.  If $Z$ is supported only on curves and points, then $\mathcal{O}_{Z}$ indeed corresponds to D2 or D0-branes, as expected.  However, in that case $\mathcal{I}_{Z}$ is a rank one torsion-free sheaf which is not locally-free.  

Ideal sheaves have no immediate interpretation as D-branes.  However, notice that because the Chern character is additive on short exact sequences, applying $\mathcal{Q}$ to (\ref{eqn:exseqidsheaf}), the D-brane charges are seen to satisfy

\begin{equation}
\mathcal{Q}(\mathcal{O}_{X}) = \mathcal{Q}(\mathcal{I}_{Z}) + \mathcal{Q}(\mathcal{O}_{Z}),
\end{equation}
which looks like a manifestation of charge conservation.  This is perhaps hinting that an ideal sheaf may have an interpretation as a bound-state of a brane ($\mathcal{O}_{X}$) and a suitably defined anti-brane ($\mathcal{O}_{Z}$) coupled via a map $\mathcal{O}_{X} \to \mathcal{O}_{Z}$.

\subsection{Summary and Outlook}

Roughly speaking, one may think of the category of coherent sheaves $\text{Coh}(X)$ as containing all of the locally-free sheaves on $X$, plus all of the ideal sheaves, structure sheaves, and pushforwards of sheaves arising from holomorphic vector bundles on subvarieties.  Thus, if we want to expand beyond the world of vector bundles, considering the coherent sheaves is the most natural first step.  We hope to argue that the derived category $D^{b}\text{Coh}(X)$ will be large enough to contain all B-model D-branes.   In the following section we will introduce some of the machinery of homological algebra and sheaf cohomology.  There are at least two indications so far that such machinery should be important.  

Recall that we have only done one computation in this section: in the case of two D6-branes, we computed the spectrum of open string states stretching between the branes.  Here we used that D6-branes correspond to vector bundles $E_{1} \to X$ and $E_{2} \to X$, and since they share a common base space, the group of morphisms $\text{Hom}(E_{1}, E_{2})$ was well-defined.  But in higher codimension, branes need not intersect, and certainly will not be given simply by a locally-free sheaf.  For example, we can have a brane supported on a divisor, and another supported on a curve with open strings stretching between.  Or we can have a stack of $N$ D0-branes supporting open string endpoints.  In this setting, it is natural to expect the \emph{Ext Groups} to encode the open string spectra, as they are a natural generalization of bundle morphisms.

In addition, the ideal sheaf short exact sequence we encountered is perhaps hinting that we should consider \emph{complexes} of coherent sheaves.  We saw that the application of the D-brane charge $\mathcal{Q}$ to such a short exact sequence seems to encode a charge conservation.  The physical BRST formalism provides a natural grading by the ghost number, so we can consider a D-brane as a direct sum, graded by the ghost number.  Turning on VEVs for a tachyon field, will deform this direct sum to a genuine complex.  The second motivation to consider complexes, comes from the general philosophy of resolutions.  It's often beneficial to replace an arbitrary element of a category by a tower of ``pleasant" objects.  In other words, you have resolved the object by a complex of nice objects.  The coherent sheaves we find to be particularly pleasant are the locally-free sheaves associated to space-filling branes.  Given a coherent sheaf which is not locally-free (coming from a D0-, D2-, or D4-brane) we can find a locally-free resolution.

Once in the category of complexes of coherent sheaves, the glaring question is, are there physical reasons to identify complexes up to homotopy and quasi-isomorphism?  Remarkably, the answer is conjecturally, yes.  Identifying homotopic maps between D-branes will be natural from the BRST formalism.  We will interpret quasi-isomorphic complexes to be in the same ``universality class" of Renormalization Group flow on the worldsheet.  Moreover, we can realize this flow as brane/anti-brane annihilation via a non-zero tachyon VEV.

\section{Sheaf Cohomology, Derived Functors, and Ext Groups}

We begin with a few remarks pertaining to the global sections of a sheaf.  We assume the reader is familiar with $\check{C}$ech cohomology.  

\begin{rmkk}
Given a sheaf $\mathscr{F}$ on $X$, the zeroth $\check{C}$ech cohomology group computes the global sections,

\[\Gamma(X, \mathscr{F}) \simeq \check{H}^{0}(X, \mathscr{F}).\]
\end{rmkk}

\begin{rmkk}
Given an $\mathcal{O}_{X}$-module $\mathscr{F}$, the global sections of $\mathscr{F}$ correspond to morphisms $\mathcal{O}_{X} \to \mathscr{F}$,

\[\Gamma(X, \mathscr{F}) \simeq \textnormal{Hom}(\mathcal{O}_{X}, \mathscr{F}).\]
\end{rmkk}
\noindent We should also record the familiar isomorphism between $\check{C}$ech cohomology and Dolbeault cohomology.  

\begin{rmkk}
Let $\Omega^{p}$ be the sheaf of holomorphic $p$-forms on $X$.  The $\check{C}$ech-Dolbeault isomorphism states that
\begin{equation}
H^{p,q}_{\bar{\partial}}(X) \simeq \check{H}^{q}(X, \Omega^{p}).
\end{equation}
More generally, we can let $E$ be a holomorphic vector bundle on $X$, with corresponding locally-free sheaf $\mathscr{E}$.  The generalized $\check{C}$ech-Dolbeault isomorphism relates $(p,q)$-forms valued in $E$ to the sheaf $\mathscr{E} \otimes \Omega^{p}$
\begin{equation}
H^{p,q}_{\bar{\partial}}(X,E) \simeq \check{H}^{q}(X, \mathscr{E} \otimes \Omega^{p}).
\end{equation}
\end{rmkk}

One important idea will be that of \emph{resolutions}.  The general philosophy of resolutions is that given an arbitrary object $A$ in some category, it might be preferable to replace $A$ by a tower of especially pleasant objects in the category.  One often speaks of injective, projective, flasque/flabby, or free resolutions.  These focus our attention on especially nice, or rigid objects in the category.  This provides a way of defining \emph{derived functors} which can be evaluated at such arbitrary objects $A$.  This can be done in some generality in the category of $R$-modules over a ring $R$.  However, we will focus on the category of $\mathcal{O}_{X}$-modules.  In this category, using resolutions to define derived functors will immediately give a definition of sheaf cohomology.  This sheaf cohomology is extremely abstract, so is not terribly helpful in explicit computations, but it agrees with $\check{C}$ech cohomology, and will allow for the definition of the Ext groups.

Given an injective resolution of some object $A$,

\begin{equation}
\begin{CD}
0 @>>> A @>>> \mathscr{I}_{0} @>>> \mathscr{I}_{1} @>>> \ldots 
\end{CD}
\end{equation}
and a left-exact functor $F$, we get a complex

\begin{equation}
\begin{CD}
0 @>>> F(\mathscr{I}_{0}) @>>> F(\mathscr{I}_{1}) @>>> F(\mathscr{I}_{2}) @>>> \ldots 
\end{CD}
\end{equation}
We define the n$^{\textnormal{th}}$ right derived functor of $F$ at $A$, denoted $\textbf{R}^{n}F(A)$, to be the n$^{\textnormal{th}}$ cohomology of the above sequence.  From here on, we will restrict attention to the category of $\mathcal{O}_{X}$-modules, with the primary left-exact functor of interest being $\textnormal{Hom}(\mathcal{O}_{X}, -)$.  By an earlier remark, we may also refer to $\textnormal{Hom}(\mathcal{O}_{X}, -)$ as the \emph{global section functor} since acting on any sheaf results in the group of global sections.  We arrive finally at the important definition of sheaf cohomology

\begin{defn}
We define sheaf cohomology for $\mathcal{O}_{X}$-modules to be the right derived functor of the left-exact global sections functor $\textnormal{Hom}(\mathcal{O}_{X}, -)$.  Given an $\mathcal{O}_{X}$-module $\mathscr{F}$, then the n$^{\textnormal{th}}$ sheaf cohomology group of $\mathscr{F}$ is
\begin{equation}
H^{n}(X, \mathscr{F}) = \textbf{R}^{n} \textnormal{Hom}(\mathcal{O}_{X}, -)(\mathscr{F}).
\end{equation}
\end{defn}
The most pressing point to be made after a definition using derived functors, is that the result is \emph{independent} of the particular resolution we chose.  Resolutions of a given object are generally far from unique, and it would clearly be problematic if we got a different result for sheaf cohomology depending on which resolution we chose; this is not the case.

\begin{rmkk}
The 0$^{\textnormal{th}}$ sheaf cohomology group computes the global sections of the sheaf.  In particular, it agrees with $\check{C}$ech cohomology.
\end{rmkk}

\begin{proof}
The proof here is very straightforward.  In general for a right derived functor, we have $\textbf{R}^{0}F(A) = F(A)$, since the functor $F$ is left-exact.  Using this, we compute

\begin{equation}
H^{0}(X, \mathscr{F}) = \textbf{R}^{0} \textnormal{Hom}(\mathcal{O}_{X}, -)(\mathscr{F})=\textnormal{Hom}(\mathcal{O}_{X}, \mathscr{F})= \mathscr{F}(X)
\end{equation}
\end{proof}
The definition of sheaf cohomology using derived functors is quite abstract.  In fact, it's so abstract, that it's essentially immune to computations.  However, this same abstraction makes it incredibly elegant to use in theory building.  When needed for actual computations, the best approach is to prove that it is isomorphic to something like $\check{C}$ech cohomology which is far more computable.

\begin{thm}
Given an $\mathcal{O}_{X}$-module $\mathscr{F}$, the $\check{C}$ech and sheaf cohomologies are isomorphic,
\begin{equation}
H^{n}(X, \mathscr{F}) \simeq \check{H}^{n}(X, \mathscr{F}).
\end{equation}
\end{thm}

\begin{proof}
I refer the reader to Theorem 4.5 in \cite{hartshorne_algebraic_1997}.
\end{proof}

\noindent We have defined sheaf cohomology for arbitrary $\mathcal{O}_{X}$-modules, however we can simply restrict attention to the coherent sheaves if we like.  The reason being, the coherent sheaves are a \emph{full} subcategory, meaning the morphisms are the same in the subcategory, as they are in the original category.  In fact, in applications to D-branes, we will usually regard the locally-free sheaves as the particularly nice objects within the category of coherent sheaves.  Thus we'll want to take an arbitrary coherent sheaf, and resolve it using a tower of locally-free sheaves.  First, we must introduce the Ext groups.

\begin{defn}
Let $\mathscr{E}$ be an $\mathcal{O}_{X}$-module.  The functor $\textnormal{Hom}(\mathscr{E}, -)$ is left-exact, so we may consider its right derived functor evaluated at an $\mathcal{O}_{X}$-module $\mathscr{F}$.  This allows for the following definition of the Ext groups,
\begin{equation}
\textnormal{Ext}^{n}(\mathscr{E}, \mathscr{F}) = \textbf{R}^{n} \textnormal{Hom}(\mathscr{E}, - ) (\mathscr{F}).
\end{equation}
\end{defn} 

\noindent There are a few simple examples where the Ext groups correspond to familiar quantities:

\begin{equation}
\textnormal{Ext}^{0}(\mathscr{E}, \mathscr{F}) = \textbf{R}^{0} \textnormal{Hom}(\mathscr{E}, -)(\mathscr{F}) = \textnormal{Hom}(\mathscr{E}, \mathscr{F}),
\end{equation}

\noindent and

\begin{equation}
\textnormal{Ext}^{n}(\mathcal{O}_{X}, \mathscr{F}) = \textbf{R}^{n} \textnormal{Hom}(\mathcal{O}_{X}, -)(\mathscr{F}) = H^{n}(X, \mathscr{F}).
\end{equation}
Thus, we see that Ext groups at least encode abelian groups of sheaf morphisms and sheaf cohomology groups.  In addition, we have the following useful result, known as Serre duality.

\begin{thm}
In the case where $X$ is a Calabi-Yau $m$-fold, for all $n=0, \ldots, m$
\begin{equation}
\textnormal{Ext}^{n}(\mathscr{E}, \mathscr{F}) \simeq \textnormal{Ext}^{m-n}(\mathscr{F}, \mathscr{E}).
\end{equation}
\end{thm}

Consider two holomorphic vector bundles $E$ and $F$ on $X$, with corresponding locally-free sheaves $\mathscr{E}$ and $\mathscr{F}$.  The space of vector bundle morphisms $\textnormal{Hom}(E, F)$ is actually itself a vector bundle, with fiber defined by $\textnormal{Hom}(E, F)(x) = \textnormal{Hom}(E(x), F(x))$, for all $x \in X$.  Since $\textnormal{Hom}(E, F)$ is a holomorphic vector bundle, there exists a corresponding locally-free sheaf which we denote as $\mathscr{H}\textnormal{om}(\mathscr{E}, \mathscr{F})$.  It is very easy to get mixed up here with the notation, so we briefly summarize,

\begin{center}
$\textnormal{Hom}(E, F)$ \,\,\,\,= \,\,\,\, \textnormal{the holomorphic vector bundle of bundle morphisms}

\hskip-7ex $\mathscr{H}\textnormal{om}(\mathscr{E}, \mathscr{F})$ \,\,\,\, = \,\,\,\, locally-free sheaf associated to $\textnormal{Hom}(E, F)$

\hskip-2ex$\textnormal{Hom}(\mathscr{E}, \mathscr{F})$ \,\,\,\, = \,\,\,\, abelian group of sheaf morphisms from $\mathscr{E}$ to $\mathscr{F}$.  
\end{center}

\noindent Moreover, $\textnormal{Hom}(\mathscr{E}, \mathscr{F})$ is actually the abelian group of global sections of $\mathscr{H}\textnormal{om}(\mathscr{E}, \mathscr{F})$.  By the $\check{C}$ech-Dolbeault-Sheaf isomorphism,

\begin{equation}
H^{0,q}_{\bar{\partial}}(X, \textnormal{Hom}(E,F)) \cong \check{H}^{q}(X, \mathscr{H}\textnormal{om}(\mathscr{E}, \mathscr{F})) \cong H^{q}(X, \mathscr{H}\textnormal{om}(\mathscr{E}, \mathscr{F})).
\end{equation}
Since $\textbf{R}^{0}\textnormal{Hom}(\mathcal{O}_{X}, -)(\mathscr{F})  = \textnormal{Hom}(\mathcal{O}_{X}, \mathscr{F})$ gives the global sections of $\mathscr{F}$, and $\textnormal{Ext}^{q}(\mathcal{O}_{X}, \mathscr{F}) \cong H^{q}(X, \mathscr{F})$, we can conclude analogously that,

\begin{equation} \label{eqn:ExtHom}
\textnormal{Ext}^{q}(\mathscr{E}, \mathscr{F}) \cong H^{q}(X, \mathscr{H}\textnormal{om}(\mathscr{E}, \mathscr{F})) \cong H^{0,q}_{\bar{\partial}}(X, \textnormal{Hom}(E, F))
\end{equation}
With this, we have finally converted all complex geometry of D-branes into algebraic and sheaf-theoretic language.  For $X$ a Calabi-Yau threefold, we conclude from (\ref{eqn:ExtHom}) and (\ref{eqn:openstringst}),

\begin{center}
\textbf{Given two stacks of D6-branes in the B-model with associated locally-free sheaves $\mathscr{E}$ and $\mathscr{F}$, the open strings states stretching from $\mathscr{E}$ to $\mathscr{F}$ with ghost number $q$, are given by the abelian group $\textnormal{Ext}^{q}(\mathscr{E}, \mathscr{F})$.}
\end{center}

\section{The Derived Category and Complexes of D-branes}

For the time being, I would like to restrict attention to D6-branes modeled as locally-free sheaves, as opposed to more general coherent sheaves.  We saw above that given two stacks of D6-branes wrapping a Calabi-Yau threefold $X$ with corresponding holomorphic vector bundles $E$ and $F$, we can ask about the morphisms between them.  These were shown to be given by the Dolbeault cohomology $H^{0,q}_{\bar{\partial}}(X, \textnormal{Hom}(E,F))$ or equivalently, $\textnormal{Ext}^{q}(\mathscr{E}, \mathscr{F})$.  We identify each morphism with a string state, and the \emph{ghost number} or R-charge of the string corresponds to $q$.  Mathematically, we can think of this $q$ as providing a natural $\mathbb{Z}$-grading.  Given a D-brane with holomorphic vector bundle $E$ (or locally-free sheaf $\mathscr{E}$), we can consider all strings attached to it as being graded by an integer.  Thus, it seems natural to initially consider direct sums of locally-free sheaves on $X$,

\begin{equation}
\mathscr{E} = \bigoplus_{n \in \mathbb{Z}} \mathscr{E}^{n}.
\end{equation}
The above direct sum can trivially regarded as a complex with all maps being zero

\begin{equation}
\begin{CD}
\dotr{\mathscr{E}} =\big( \ldots @>0>> \mathscr{E}^{-1} @>0>> \mathscr{E}^{0} @>0>> \mathscr{E}^{1} @>0>> \ldots \big).
\end{CD}
\end{equation}
Simply put, we want to deform away from the trivial case of direct sums by turning on non-zero maps between the $\mathscr{E}^{i}$ in the above sequence.  These non-zero maps will be called \emph{tachyons} for reasons to be explained shortly.  Once we do this, the D-branes will correspond to elements in the category of complexes $\textbf{Kom}(\mathcal{C})$, where $\mathcal{C}$ is the category of locally-free sheaves on $X$.  However, physically, the string states correspond to elements in $Q$-cohomology, so we need to identify all states differing by a $Q$-exact terms.  Remarkably, this identification on the physics side, corresponds precisely to identifying complexes up to homotopy in $\textbf{Kom}(\mathcal{C})$.  This places the B-model D-branes in correspondence with the homotopy category $\bf{K}(\mathcal{C})$.

This correspondence is certainly elegant, but there is a fundamental problem here.  Most importantly, the homological algebra described just above requires that the category be \emph{abelian}.  The category of locally-free sheaves is additive, but not abelian.  The resolution here will be to extend our consideration to the category of coherent sheaves $\text{Coh}(X)$, which is an abelian category containing the category of locally-free sheaves $\mathcal{C}$.  This seemingly dangerous problem was actually hinting that we weren't considering all of the branes that we need to.  As we saw in an earlier section, the locally-free sheaves cannot describe D4, D2, nor D0-branes; these require torsion sheaves.  Thus, extending to coherent sheaves is well-motivated both mathematically and physically.

Given the homotopy category $\textbf{K}\text{Coh}(X)$ of coherent sheaves, it is tempting to identify quasi-isomorphisms and arrive at the (bounded) derived category $D^{b}\text{Coh}(X)$.  But is there any physical motivation for this?  Indeed, we will argue that two complexes which are quasi-isomorphic lie in the same universality class of the renormalization group flow.  In other words, one complex can be thought of as \emph{condensing} to another \cite{sen_tachyon_1998}.  This explains the use of the term tachyon: in string theory, a tachyon is a particle which signifies an instability.  This instability corresponds to the branes in a complex annihilating each other.

\subsection{Deformation of Complexes} 

Let us begin by considering a stack of D6-branes on a threefold $X$ given by a holomorphic vector bundle $E$ (with associated sheaf $\mathscr{E}$), which decomposes as the direct sum

\begin{equation}
\mathscr{E} = \bigoplus_{n \in \mathbb{Z}} \mathscr{E}^{n},
\end{equation}
where each $\mathscr{E}^{i}$ is a locally-free sheaf on $X$.  The open string states from $\mathscr{E}$ to itself, correspond to linear combinations of elements in $\textnormal{Ext}^{*}(\mathscr{E},\mathscr{E})$.  For all $n$, $k$ the string states with ghost number $q$ correspond to elements of $\textnormal{Ext}^{k}(\mathscr{E}^{n}, \mathscr{E}^{n-k + q})$.  For example, when $k=1$, the string with ghost number $q=1$ correspond to elements in $\textnormal{Ext}^{1}(\mathscr{E}^{n}, \mathscr{E}^{n})$, which describe deformations of the locally-free sheaf $\mathscr{E}$ associated to the vector bundle $E$.  The more pressing case to consider is $k=0$.  Here, the ghost number $q=1$ strings are elements of $\textnormal{Ext}^{0}(\mathscr{E}^{n}, \mathscr{E}^{n+1}) \simeq \textnormal{Hom}(\mathscr{E}^{n}, \mathscr{E}^{n+1})$.  Let us define $d=\sum d_{n} \in \textnormal{Hom}(\mathscr{E}, \mathscr{E})$, where

\begin{equation}
d_{n} \in \textnormal{Hom}(\mathscr{E}^{n}, \mathscr{E}^{n+1}).
\end{equation}
Thus, $d_{n}$ is a holomorphic map from $\mathscr{E}^{n}$ to $\mathscr{E}^{n+1}$.  We can use $d$ to deform the physical sigma model action

\begin{equation}
\delta S = \oint_{\partial \Sigma} (\psi_{+}^{i} + \psi_{-}^{i})\partial_{i}d,
\end{equation}
and then prove that deforming the action by $\delta S$, requires a deformation of the BRST operator as well

\begin{equation}
Q= Q_{0} + d.
\end{equation}
We need to retain the nilpotence $Q^{2} =0$ of the BRST operator, which leads to the constraint

\begin{equation}
\{Q_{0}, d\} + d^{2} =0.
\end{equation}
The two terms above must individually vanish.  The constraint $\{Q_{0}, d\}=0$ is merely the statement that $d$ is a holomorphic map, recalling that in the B-model, the undeformed BRST operator is $Q_{0} = \bar{\partial}$.  The condition $d^{2}=0$ can be expanded in terms of successive maps, and we see that

\begin{equation}
d_{n+1} d_{n} =0.
\end{equation}
Thus, the nilpotence of the deformed BRST operator $Q$ is translated into the conditions that each $d_{n}$ must be a holomorphic map, and the consecutive application of two successive maps must vanish.  That is to say $\mathscr{E}$ is deformed into the complex

\begin{equation}
\begin{CD}
\dotr{\mathscr{E}} = \big(\ldots @>d_{n-1}>> \mathscr{E}^{n} @>d_{n}>> \mathscr{E}^{n+1} @>d_{n+1}>> \mathscr{E}^{n+2} @>d_{n+2}>> \ldots\big)
\end{CD}
\end{equation}

Consider now the slightly more general case of open strings stretching from a stack of D6-branes $\dotr{\mathscr{E}}$ to another stack of D6-branes $\dotr{\mathscr{F}}$.  Let both $\dotr{\mathscr{E}}$ and $\dotr{\mathscr{F}}$ consist of a collection of objects (graded by ghost number) constituting a trivial complex; namely, $\dotr{\mathscr{E}}$ and $\dotr{\mathscr{F}}$ decompose into direct sums.  We deform the theory by turning on the differentials $d^{\mathscr{E}}$ and $d^{\mathscr{F}}$, yielding two  non-trivial complexes.  The deformed BRST operator can be shown to be

\begin{equation}
Q = Q_{0} + d^{\mathscr{E}} - d^{\mathscr{F}}.
\end{equation}
Let $f^{n}: \mathscr{E}^{n} \to \mathscr{F}^{n}$ be a collection of maps from the elements of the complex $\dotr{\mathscr{E}}$, to the complex $\dotr{\mathscr{F}}$.  These should be thought of intuitively as strings stretching from $\mathscr{E}^{n}$ to $\mathscr{F}^{n}$.  What are the conditions that the map of complexes $f$ is BRST invariant?  We require 

\begin{equation}
Qf^{n}= Q_{0}f^{n} + f^{n+1}d^{\mathscr{E}} - d^{\mathscr{F}} f^{n}=0.
\end{equation}
Like above, this factors into two independent constraints.  First, we require $Q_{0}f^{n}=0$ for all $n$.  That is to say $f^{n}$ is a holomorphic map, $f^{n} \in \textnormal{Hom}(\mathscr{E}^{n}, \mathscr{F}^{n})$.  The second condition is that $f^{n+1}d^{\mathscr{E}}=d^{\mathscr{F}} f^{n}$.  This is precisely what it means for $f$ to define a morphism of complexes.  Moreover, if two such maps $f$ and $f'$ differ by a $Q$-exact term,

\begin{equation}
f' = f + Qh,
\end{equation}
then we see that $f$ and $f'$ are homotopic morphisms of complexes.  Therefore, quotienting by homotopy equivalence is the mathematical manifestation of passing to $Q$-cohomology, which is well-motivated physically.  In order that a map $f: \dotr{\mathscr{E}} \to \dotr{\mathscr{F}}$ be a genuine morphism of complexes, we require that $f$ be BRST invariant ($Q$-closed) which is precisely the physical notion of corresponding to an allowed open string state.  Moreover, two such states $f$ and $f'$ are deemed physically equivalent if and only if they differ by a $Q$-exact term, and this coincides with the definition of $f$ and $f'$ being homotopic chain maps!  We conclude that the homotopy category $\textbf{K}(\mathcal{C})$ naturally models stacks of D6-branes in the B-model.

\subsection{Renormalization Group (RG) Flow and Quasi-Isomorphisms}

Given the homotopy category $\textbf{K}(\mathcal{C})$, it is clearly tempting to ask if there is any physical motivation to identify quasi-isomorphisms, landing us once and for all in the derived category.  I hope to outline the state-of-the-art conjecture that D-branes in the B-model related by a quasi-isomorphism correspond to physical configurations related by worldsheet renormalization group (RG) flow; in some loose sense, the branes and anti-branes at least partially annihilate.

\subsection*{Branes, Anti-Branes, and Tachyons}

First, we introduce some terminology.  Given a D-brane represented as a complex,

\begin{equation}
\begin{CD}
\dotr{\mathscr{E}} = \big(\ldots @>d_{n-1}>> \mathscr{E}^{n} @>d_{n}>> \mathscr{E}^{n+1} @>d_{n+1}>> \mathscr{E}^{n+2} @>d_{n+2}>> \ldots\big)
\end{CD}
\end{equation}
we consider the entries of the complex to be alternating branes and anti-branes, and we call the maps $d_{n}$ \textit{tachyons}.  In string theory, a tachyon indicates an instability in a physical system.  Indeed, here we mean that non-zero tachyons $d_{n}$ lead to an instability of the configuration of D-branes.  In the trivial case where all $d_{n}$ vanish, the system appears to be in a stable state, but with non-zero tachyons, the configuration may flow via the renormalization group to a more stable system.

In physics, two theories related by renormalization group flow are said to lie in the same \textit{universality class}.  The conjecture here is that the physical universality classes correspond to the equivalence classes of quasi-isomorphisms.  It's crucial to note that two theories in the same universality class, are not equivalent physical theories.  RG flow represents a flow in the ``space of theories" to a completely different physical theory.

We seem to have argued that the category of D-branes in the B-model topological string is the bounded derived category $D^{b}(\mathcal{C})$ of \emph{locally-free sheaves} on $X$.  However, as mentioned earlier, the locally-free sheaves are not an abelian category and this is not well-defined.  The natural guess is to pass to the coherent sheaves which are essentially the abelianization of the locally-free sheaves.  Indeed, we have seen that a D-brane na\"{i}vely corresponds to a coherent sheaf, making this extension reasonable.  Therefore, the conjectural conclusion provided by \cite{aspinwall_d-branes_2004,sharpe_lectures_2003} is,

\begin{center}
\textbf{The category of D-branes in the B-model topological string are the $\Pi$-stable\footnote{The $\Pi$-stability of M. Douglas was introduced in \cite{douglas_stability_2005}, see also \cite{aspinwall_d-brane_2002}.  This was the physical precursor to Bridgeland stability on a triangulated category \cite{bridgeland_stability_2007}.} objects in the derived category $\bm{D^{b}\text{Coh}(X)}$ of coherent sheaves on a Calabi-Yau threefold $\bm{X}$.  Quotienting by homotopy corresponds to identifying states up to $Q$-exact terms.  Quasi-isomorphism corresponds to worldsheet RG flow and brane/anti-brane annihilation.}    
\end{center}

\noindent This is fundamentally built on the original conjecture of Kontsevich \cite{kontsevich_homological_1994} relating homological algebra and mirror symmetry.

\section{Examples}

\subsection*{Example 1: Elementary Brane/Anti-Brane Annihilation}

The following example is as simple as it gets, but it illustrates well all of the features of the discussion above.  Consider the following complex,

\begin{equation}
\begin{CD}
\ldots @>>> 0 @>>> E @>c>> E @>>>0 @>>> \ldots
\end{CD}
\end{equation}
where $E$ can be any coherent sheaf supported on a subvariety $Z$ of $X$.  If the above map $c$ is identically zero, then we essentially have two copies of $E$ which do not couple in any sense, and the complex decomposes into a direct sum $E \oplus E$.  This can be thought of as two stacks of D-branes wrapped on $Z$ which do not interact.  If, however, we turn on the map $c \neq 0$, physically, we have added a VEV for a tachyon field, which indicates an unstable coupling between the branes.  The $E$ on the left represents an anti-brane while the $E$ on the right represents a brane.  Intuitively, we physically expect the branes to annihilate due to this instability.  In other words, this sequence should be in the same universality class as the zero complex.  Indeed, in this case this sequence is quasi-isomorphic to its cohomology, which is simply zero in every entry, for $c \neq 0$.  And so the intuition is verified: the unstable configuration is quasi-isomorphic to the zero complex, signifying brane/anti-brane annihilation.

\subsection*{Example 2: D4-Branes}

Let $Z \subseteq X$ be a codimension one complex subvariety of $X$, i.e. a divisor.  This means $Z$ is cut out by a section of a line bundle $\mathcal{O}(-Z)$, which is given locally by the vanishing of a holomorphic function $f$.  Since $\mathcal{O}(-Z)$ consists simply of all holomorphic functions vanishing identically on $Z$, it naturally injects into $\mathcal{O}_{X}$.  The cokernel of this map $\mathcal{O}(-Z) \to \mathcal{O}_{X}$ is simply the structure sheaf of $Z$, denoted $\mathcal{O}_{Z}$.  Thus, we have the following short exact sequence of sheaves,

\begin{equation}
\begin{CD}
0 @>>>\mathcal{O}(-Z) @>f>> \mathcal{O}_{X} @>>> \mathcal{O}_{Z} @>>> 0.
\end{CD}
\end{equation}
We can reinterpret this exact sequence.  Let us define the complex (\emph{not} exact sequence) $\dotr{\mathcal{E}}$,

\begin{equation}
\begin{CD}
\dotr{\mathscr{E}} = \big(\ldots @>>> 0 @>>>\mathcal{O}(-Z) @>f>> \mathcal{O}_{X} @>>> 0 @>>>\ldots \big),
\end{CD}
\end{equation}
where we take $\mathcal{O}_{X}$ to be in the zeroth degree slot.  It's straightforward to compute the cohomology of this complex, where we get zero in all degrees except $\mathscr{H}^{0}(\dotr{\mathscr{E}}) = \mathcal{O}_{X}/\mathcal{O}(-Z) = \mathcal{O}_{Z}$.  Since we are working in codimension one, $\mathcal{O}(-Z)$ is a locally-free sheaf, so we see that we have recovered the coherent sheaf $\mathcal{O}_{Z}$ as the cohomology of a complex of locally-free sheaves.  Thus, we can interpret the complex $\dotr{\mathscr{E}}$ as consisting of a brane and an anti-brane annihilating to yield simply $\mathcal{O}_{Z}$ as the endpoint of renormalization group flow.  It's natural to consider $\mathcal{O}_{Z}$ as associated to a D4-brane, since it is supported on a four real dimensional manifold $Z$.  

Interpreting this example in another light, we regard $\mathcal{O}(-Z) \to \mathcal{O}_{X}$ as a locally-free resolution of the torsion coherent sheaf $\mathcal{O}_{Z}$.  Of course, the cohomology of a resolution coincides with the original object, itself.  In this way, we can imagine resolving any coherent sheaf supported on a complex subvariety by locally-free sheaves.  This is what we meant above, when we mentioned that coherent sheaves arise naturally from complexes of locally-free sheaves, under RG flow.

\subsection*{Example 3: D0-Branes}

Let us take $X$ to be a Calabi-Yau threefold which we can study locally as $\mathbb{C}^{3}$ with coordinates $(x,y,z)$.  We define a map

\begin{equation}
\begin{CD}
\mathcal{O}_{X}^{\oplus 3} @> (x \, y\,  z)>> \mathcal{O}_{X},
\end{CD}
\end{equation}
defined by taking a triple of holomorphic functions $(f_{1}, f_{2}, f_{3})$ to the holomorphic function $xf_{1} + yf_{2} + zf_{3}$.  Since the cokernel of this map should be $\mathcal{O}_{X}$ modulo the image of this map, we expect that a section of the cokernel must vanish away from the origin in $\mathbb{C}^{3}$, but can take any complex value at the origin.  Letting $p$ denote the origin in $\mathbb{C}^{3}$, we see that the cokernel of the above map is simply the skyscraper sheaf $\mathcal{O}_{p}$ at $p$.  Moreover, we naturally have a surjective map $\mathcal{O}_{X} \to \mathcal{O}_{p}$ arising from evaluating a holomorphic function at $p$.  Thus, the skyscraper sheaf $\mathcal{O}_{p}$ is a coherent sheaf.  We have the exact sequence

\begin{equation}
\begin{CD}
\mathcal{O}_{X}^{\oplus 3} @> (x \, y\,  z)>> \mathcal{O}_{X} @>>> \mathcal{O}_{p} @>>> 0. 
\end{CD}
\end{equation}
Indeed, this map has a kernel, corresponding to the sheaf of all functions vanishing at $p$; this is the ideal sheaf of $p$.  Finally, this gives the short exact sequence

\begin{equation}
\begin{CD}
0 @>>> \mathcal{I}_{p} @>>> \mathcal{O}_{X} @>>> \mathcal{O}_{p} @>>> 0. 
\end{CD}
\end{equation}
Just like in the previous example, we can define the complex $\dotr{\mathscr{E}}$,

\begin{equation}
\begin{CD}
\dotr{\mathscr{E}} = \big(\ldots @>>> 0 @>>>\mathcal{I}_{p} @>>> \mathcal{O}_{X} @>>> 0 @>>>\ldots \big),
\end{CD}
\end{equation}
and trivially compute the cohomology of the complex to vanish in all degrees except $\mathscr{H}^{0}(\dotr{\mathscr{E}}) = \mathcal{O}_{X}/\mathcal{I}_{p} = \mathcal{O}_{p}$.  Once again, we recover a general coherent sheaf as the cohomology of a complex of coherent sheaves.  The only difference here is that $\mathcal{I}_{p}$ is no longer a locally-free sheaf.  Rather, it can be though of roughly as a trivial line bundle outside the origin, where there is no fiber.

\subsection*{Example 4: Branes Wrapping Curves and Points}

The setting most closely aligned with modern enumerative geometry and string theory consists of studying one-dimensional sheaves on a Calabi-Yau threefold $X$.  These are sheaves which have a complex one-dimensional support.  The Gromov-Witten, Donaldson-Thomas, and Gopakumar-Vafa invariants often package themselves into partition functions exhibiting remarkable properties, and uncovering surprising connections to subjects like modular forms, and representation theory, to name just a few.  As a final example, I would like to briefly outline the connection two of these invariants have with the content of this article.  

The \emph{Donaldson-Thomas invariants} are a (virtual) count of subschemes $Z \subseteq X$ supported on a fixed homology class $\beta \in H_{2}(X, \mathbb{Z})$. and whose structure sheaf $\mathcal{O}_{Z}$ has a fixed holomorphic Euler characteristic.  Such subschemes can be supported on both curves and points.  We therefore think of $\mathcal{O}_{Z}$ as a \emph{bound state of D2-D0 branes}.  However, there is necessarily a surjective map $\mathcal{O}_{X} \to \mathcal{O}_{Z}$, with kernel $\mathcal{I}_{Z}$.  Therefore, one often hears the Donaldson-Thomas invariants described as enumerating \emph{bound states of D2-D0 (anti) branes within a single D6-brane}.

The \emph{Gopakumar-Vafa invariants} are integers which count BPS states of D2-branes wrapped on curves in $X$.  Given a class $\beta \in H_{2}(X, \mathbb{Z})$ we can consider the moduli space $\mathcal{M}(0,0,\beta, 1)$ of pure one-dimensional sheaves $\mathscr{F}$ supported on class $\beta$ with D-brane charge $\mathcal{Q}(\mathscr{F}) = (0,0,\beta,1)$.  Recently, a proposal emerged \cite{maulik_gopakumar-vafa_2016} for extracting the Gopakumar-Vafa invariants from $\mathcal{M}(0,0,\beta, 1)$ consistent with their known relation to Gromov-Witten theory.

\vskip5ex
\emph{Due to the combined shortness of my talk, and immense breadth of this subject, I necessarily had to omit certain important topics and examples.  Particularly, some explicit computations of open string states as Ext groups.  These are covered excellently in \cite{sharpe_lectures_2003, aspinwall_d-branes_2004} to which I refer the reader.  In particular, \cite{sharpe_lectures_2003} contains quite a few very enlightening examples.  Another topic I had to omit is spectral sequences; a great discussion can be found in \cite{aspinwall_d-branes_2004}.}

\subsection*{Acknowledgements}

I am extremely grateful for the opportunity to attend the 2016 Superschool on Derived Categories and D-branes.  I'd like to thank the organizers Matthew Ballard, Charles Doran, and David Favero.  First and foremost, I am especially indebted to Eric Sharpe.  Without Eric's immense patience and lengthy responses to my questions, I would not have had the opportunity to understand the physical aspects of this subject.  I'd also like to thank all the attendees for an excellent week, particularly Jake Bian who provided valuable input and discussion.  Finally, I am thankful to my advisor Jim Bryan for reading a draft and offering extremely helpful suggestions.

\bibliographystyle{siam}
\bibliography{MyLibrary}
\end{document}